\documentclass[prl,,aps,superscriptaddress,,twoside,twocolumn]{revtex4}

\usepackage{amsfonts,amssymb,amsmath,amsthm,graphicx}

\newtheorem{theorem}{Theorem}

\newtheorem{corollary}{Corollary}

\newcommand*{\bbC}{\mathbb{C}}
\newcommand*{\complex}{\bbC}

\newcommand*{\cE}{\mathcal{E}}

\newcommand*{\cL}{\mathcal{L}}

\newcommand*{\cS}{\mathcal{S}}

\newcommand*{\conv}{\mathrm{conv}}

\newcommand*{\id}{\openone}

\newcommand*{\tr}{\mathrm{tr}}
\newcommand*{\ket}[1]{| #1 \rangle}
\newcommand*{\bra}[1]{\langle #1 |}

\newcommand*{\proj}[1]{\ket{#1}\bra{#1}}

\newcommand*{\eps}{\varepsilon}
\newcommand*{\Sym}{\mathrm{S}}

\newcommand*{\LOCC}{\text{LOCC}}

\begin{document}

\title{Detection of Multiparticle Entanglement: \\ Quantifying the Search for Symmetric Extensions}


\author{Fernando~G.S.L.~Brand\~ao}
\affiliation{Departamento de F\'isica, Universidade Federal de Minas Gerais, Belo Horizonte, Brazil.}

\author{Matthias Christandl}
\affiliation{Institute for Theoretical Physics, ETH Zurich, Wolfgang-Pauli-Strasse 27, CH-8057 Zurich, Switzerland.}


\begin{abstract}
We provide quantitative bounds on the characterisation of multiparticle separable states by states that have locally symmetric extensions. The bounds are derived from two-particle bounds and relate to recent studies on quantum versions of de Finetti's theorem. We discuss algorithmic applications of our results, in particular a quasipolynomial-time algorithm to decide whether a multiparticle quantum state is separable or entangled (for constant number of particles and constant error in the LOCC or Frobenius norm). Our results provide a theoretical justification for the use of the Search for Symmetric Extensions as a practical test for multiparticle entanglement.
\end{abstract}

\maketitle

Entanglement between two particles is a fundamental resource in quantum communication theory, being of vital importance in quantum teleportation~\cite{BBCJPW93}, quantum key distribution~\cite{BB84, Eke91} as well as more exotic tasks such as the simulation of noisy channels by noiseless ones~\cite{BDHSW10, BCR11}. The most famous criterion to decide whether or not a state is entangled is the Peres-Horodecki test \cite{Per96, HHH96}: it is based on the observation that the partial transpose of a separable state is positive semi-definite and hence, if the partial transpose of a quantum state $\rho_{AB}$ is not positive semi-definite, then $\rho_{AB}$ must be entangled. Unfortunately, this criterion is only complete for two-by-two and two-by-three dimensional systems thanks to the famous entangled PPT states \footnote{A state is called PPT if its partial transpose is positive semidefinite.}~\cite{HHH98}. 

A hierarchy of separability criteria that detects every entangled state is the Search for Symmetric Extensions~\cite{DPS04}. This hierarchy is based on the observation that if $\rho_{AB}$ is separable, i.e.~of the form $\sum_i p_i \proj{\phi_i}_A\otimes \proj{\psi_i}_B$, then for every $k$, we can define the state $\sum_i p_i \proj{\phi_i}_A^{\otimes k} \otimes \proj{\psi_i}_B$ which is manifestly symmetric under the permutation of the $A$ systems and extends the original state $\rho_{AB}$~\footnote{We say that a state $\rho_{CD}$ extends $\rho_{C}$ if $\tr_D \rho_{CD}=\rho_C$.}. Hence if for some $k$ a given state $\rho_{AB}$ does not have an extension to $k$ copies of $A$ that is symmetric under interchange of the copies of $A$, then it must be entangled. The $k$'th separability criterion is thus the search for a symmetric extension to $k$ copies of $A$. Quantum versions of the famous de Finetti theorem from statistics show that this hierarchy of criteria is complete~\cite{Sto69, HM76, RW89, Wer89, CFS08} --- i.e.~every entangled state fails to have a symmetric extension for some $k$ --- and even provide quantitative bounds for the distance to the set of separable states measured in the trace norm~\cite{KR05, CKMR07} (see Figure~\ref{fig:sets}).

Interestingly, these bounds can be improved if we restrict the hierarchy to PPT states as has been shown in~\cite{NOP09a, NOP09b} following a proposal to use the search for such extensions by semidefinite programming as a test to detect bipartite entanglement~\cite{DPS04}. Whereas the algorithm works well in practice, from the bounds one can only infer a runtime exponential in the dimension the state, suggestively in agreement with the well-known result that the separability problem is NP-hard~\cite{Gur04, Gha10}. In recent work, we have shown together with Jon Yard that the algorithm runs in quasipolynomial time (even without the PPT constraints) for constant error when one is willing to consider the weaker LOCC or Frobenius norms~\cite{BCY10a, BCY10b}. The LOCC norm is an operationally defined norm giving the optimal probability of distinguishing two two-particle states by local operations and classical communication~\footnote{An algorithm has an error of at most $\varepsilon$ if it identifies all separable states correctly as well as all states that have at least distance $\varepsilon$ to the set of separable states in the chosen norm.}. 

Following their work in the two-particle case, Doherty et al.~proposed a similar search for extensions in order to detect multiparticle entanglement~\cite{DPS05}. With this Letter we provide a quantitative analysis of this proposal. We do this by deriving a bound on the distance between multiparticle states that have symmetric extensions and multiparticle separable states in terms of the corresponding two-particle bounds. We illustrate the analysis by considering the best known two-particle bounds for the norms mentioned above. As in the two-particle case, the LOCC and Frobenius norm result is shown to imply a \textit{quasipolynomial-time algorithm} for the detection of entanglement for a constant number of parties and for constant error. The practical use of the Search for Symmetric Extensions as multiparticle entanglement criteria has therefore been given a theoretical underpinning with this work. 

We also show how our results can lead to novel quantum versions of de Finetti's theorem. Whereas we do not obtain new insights for the trace norm, we obtain a de Finetti theorem in the LOCC norm that depends only logarithmically on the local dimension. This stands in sharp contrast to the trace norm case, where the dependence on the local dimension is at least linear~\cite{ CKMR07}.

The Letter is structured as follows. First we introduce the notation needed to study symmetric extensions of multiparticle quantum states. We then provide the quantitative analysis of the search for extensions as a criterion for entanglement and separability and discuss algorithmic applications. Subsequently, we derive the novel quantum de Finetti theorem before concluding this Letter with a discussion of the practical implications of this work.

\begin{figure}
\includegraphics[width=0.5\textwidth]{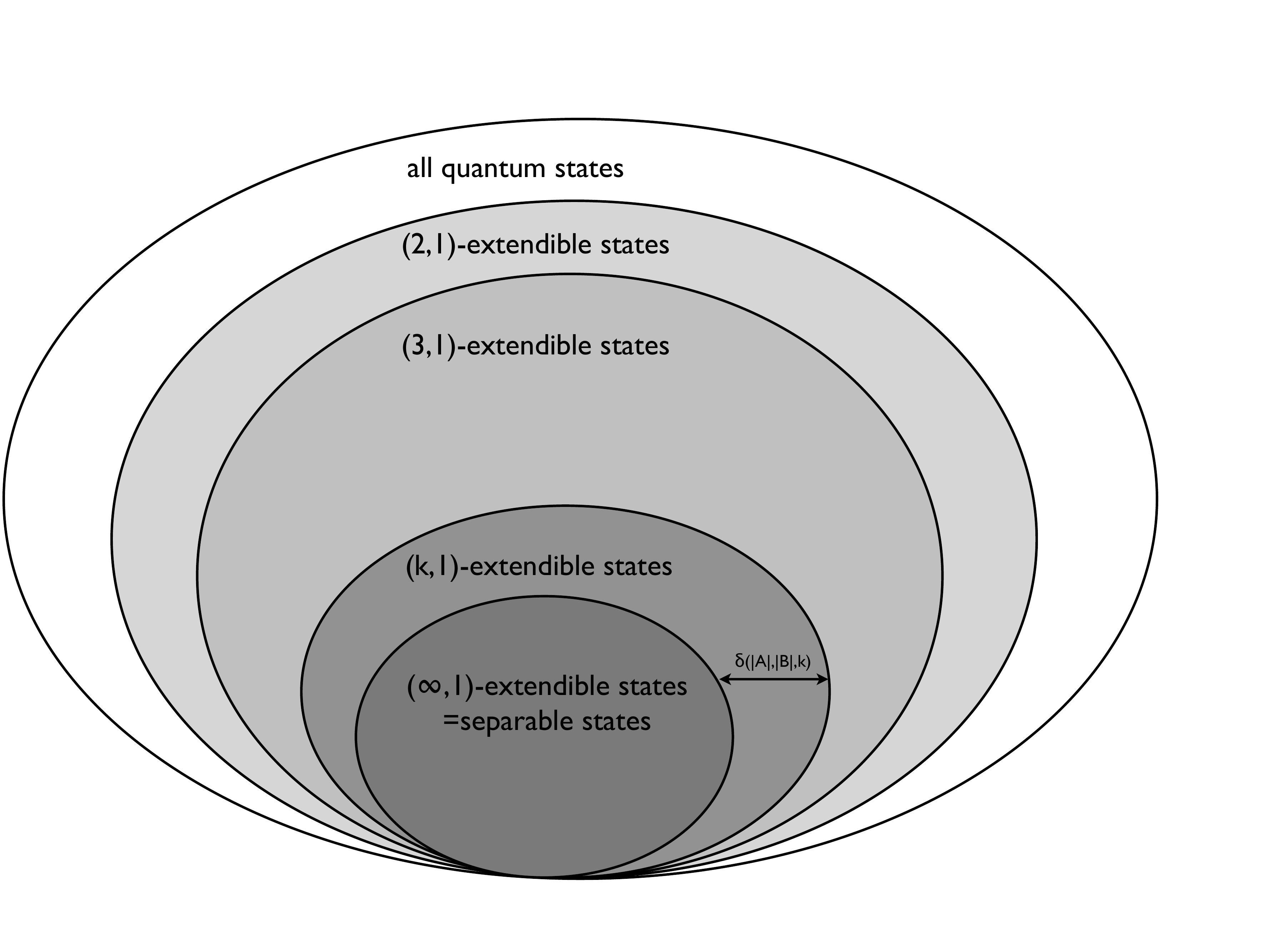} 
\caption{Illustration of the hierarchy for two particles}
\label{fig:sets}
\end{figure}

\vspace{0.3 cm}

\noindent \textbf{Symmetric Extensions:} We denote by $A_1, A_2, \cdots $ Hilbert spaces of finite but possibly different dimension $|A_i|$. We let $\cS_{A_1:A_2: \cdots :A_N}:= \conv \{ \proj{\phi_1}_{A_1} \otimes \proj{\phi_2}_{A_2 }   \cdots  \otimes \proj{\phi_N}_{A_N}\}$ be the set of separable states, where $\conv$ denotes the convex hull. We also define the convex sets of symmetrically extendible states $\cE_{A_1:A_2:  \cdots :A_N}^{k_1, k_2, \cdots k_N}$ consisting of all $\rho_{A_1 \cdots A_N}$ for which there is a state $\rho_{\Sym^{k_1}(A_1)\cdots \Sym^{k_N}(A_N)}$ with 
$$\rho_{A_1 \cdots A_N}=\tr_{A^{k_1-1}_1 \cdots A^{k_N-1}_N} \left( \rho_{\Sym^{k_1}(A_1)\cdots \Sym^{k_N}(A_N)} \right).$$
Here, $\Sym^k(A)$ denotes the symmetric subspace of $A^k \equiv A^{\otimes k}$ and $\tr_{A^{k-1}}$ stands for the partial trace of all but one of the $A$ systems~\footnote{For a permutation $\pi$ of $k$ elements define $U_\pi$ by $U_\pi \ket{i_1}\ket{i_2} \cdots \ket{i_k} = \ket{i_{\pi^{-1}(1)}}\ket{i_{\pi^{-1}(2)}}\cdots \ket{i_{\pi^{-1}(k)}} $.  $\Sym^k(A)$ then equals the vector space consisting of all $\ket{\psi} \in A^{\otimes k}$ that satisfy $U_\pi \ket{\psi}=\ket{\psi}$ for all $\pi$.}. 

In order to measure distances between quantum states we consider a norm $||* ||$ that is defined for all spaces of linear operators $\cL(A_1  \otimes \cdots \otimes A_N)$ and that may depend on the decomposition into tensor factors (here indicated by colons) satisfying the following compatibility conditions: For all finite dimensional $A_1, A_2, \cdots, A_N, A_1', A_2', \cdots, A_N', A'$ and for all completely positive trace preserving maps $\Lambda_i: \cL(A_i) \rightarrow  \cL(A_i')  $ we have
\begin{align} 
||\Lambda_1 \otimes \Lambda_2 \otimes \cdots \otimes  \Lambda_N (*)  & ||_{A'_1:A'_2: \cdots : A'_N } \nonumber \\ 
\label{norm}
&  \leq  || * ||_{A_1:A_2: \cdots : A_N} .
\end{align}
and 
\begin{align} \label{norm2} ||* ||_{A':A_1:A_2: \cdots : A_N } \leq  || * ||_{A'A_1:A_2: \cdots : A_N} .
\end{align}

An example of a norm which satisfies the two conditions is the trace norm which can be written in the form $$\Vert X \Vert_1= \max_{0 \leq M \leq \id} \tr((2M - \id) X).$$ 
Note that it is independent of the split of the total Hilbert space into tensor products. 
A second norm satisfying the conditions is the LOCC norm, defined in analogy with the trace norm as
\begin{equation*}
\Vert X \Vert_{\LOCC} := \max_{M \in \LOCC} \tr((2M - \id) X),
\end{equation*}
where $\LOCC$ is the convex set of matrices $0 \leq M \leq \id$ such that there is a two-outcome measurement $\{M, \id - M \}$ that can be realized by LOCC. Note that the LOCC norm does depend on the tensor product split.

We say that $\delta\equiv \delta(|A|, |B|, k)$ is a \emph{two-particle bound} for a norm $||* ||$ if for all $\rho_{AB} \in \cE^{k, 1}_{A:B}$ there exists $\sigma \in \cS_{A:B}$ with  (see Figure~\ref{fig:sets})
$$||\rho-\sigma||_{A:B}\leq \delta(|A|, |B|, k). $$
Note that $\delta(|A|, |B|, k)$ does not equal $\delta(|B|, |A|, k)$ in general. In fact, all known bounds either depend only on $|A|$ or only on $|B|$. \newline

\noindent \textbf{Main Results:} We derive two results that quantify the closeness of a separable state to a symmetrically extendible multiparticle state in terms of two-particle bounds. The first result is tailored to a two-particle bound that only depends on the dimension of $A$. For notational simplicity, we disregard the dimension of $B$ by upper bounding it with $\infty$.

\begin{theorem} \label{thm:main}
Let $||*||$ be a norm that satisfies \eqref{norm} and \eqref{norm2} and assume that $\delta(|A|, |B|, k)$ is a two-particle bound for $||*||$. Then for all $\rho\in  \cE_{A_1:A_2: \cdots :A_N}^{k_1, k_2, \ldots, k_{N}}$ there exists $\sigma \in \cS_{A_1:A_2: \cdots :A_N}$ with
$$||\rho -\sigma||_{A_1:A_2: \cdots : A_N} \leq \sum_{i=1}^{N-1} \delta \left(|A_i|,\infty, k_i\right).$$
\end{theorem}
\begin{proof} 
By assumption there exists an extension $\rho_{\Sym^{k_1}(A_1)\Sym^{k_2}(A_2)\cdots \Sym^{k_N}(A_N)}$ of $\rho_{A_1A_2\cdots A_N}$. Since clearly $\rho_{A_1\Sym^{k_2}(A_2)\cdots \Sym^{k_N}(A_N)}\in \cE_{A_1:\Sym^{k_2}(A_2)\cdots \Sym^{k_N}(A_N)}^{k_1,1} $ there exists a state $ \sigma_{A_1\Sym^{k_2}(A_2)\cdots \Sym^{k_N}(A_N)}$ of the form
\begin{align*}
 \sigma_{A_1\Sym^{k_2}(A_2)\cdots \Sym^{k_N}(A_N)}=\sum_{i_1} p_{i_1} \chi^{i_1}_{A_1}\otimes \rho^{i_1}_{\Sym^{k_2}(A_2)\cdots \Sym^{k_N}(A_N)}
 \end{align*}
  such that 
\begin{align*}
||  \rho -\sigma  ||_{A_1:\Sym^{k_2}(A_2)\cdots \Sym^{k_N}(A_N) }    \leq  \delta (|A_1|, \infty, k_1).
\end{align*}
We now apply the same reasoning to each of the $ \rho^{i_1}_{\Sym^{k_2}(A_2)\cdots \Sym^{k_N}(A_N)}$ and find that there are states
 $$\sigma^{i_1}_{A_2\Sym^{k_3}(A_3)\cdots \Sym^{k_N}(A_N)}=\sum_{i_2} p_{i_2|i_1} \chi^{i_1i_2}_{A_2}\otimes \rho^{i_1i_2}_{\Sym^{k_3}(A_3)\cdots \Sym^{k_N}(A_N)}$$ satisfying
\begin{align*}
||\rho^{i_1}- \sigma^{i_1}  & ||_{A_2:\Sym^{k_3}(A_3)\cdots \Sym^{k_N}(A_N) } \leq  \delta (|A_2|, \infty, k_2).
\end{align*}
We continue this way until
\begin{align*}
||\rho^{i_1i_2 \cdots i_{N-2}}-  \sigma^{i_1i_2 \cdots i_{N-2}} &  ||_{A_{N-1}: \Sym^{k_N}(A_N)}\\
& \leq  \delta (|A_{N-1}|, \infty,k_{N-1})
\end{align*}
for 
\begin{align*}
\sigma^{i_1i_2 \cdots i_{N-2}}_{A_{N-1}\Sym^{k_N}(A_N)} & =\sum_{i_{N-1}} p_{i_{N-1}|i_1i_2 \cdots i_{N-2}}   \\
& \quad \times \chi^{i_1i_2\cdots i_{N-1}}_{A_{N-1}} \otimes \rho^{i_1i_2 \cdots i_{N-1}}_{\Sym^{k_N}(A_N)}.  
\end{align*}
In all estimates we now take the partial trace over the remaining extensions and append $\chi_{A_1}^{i_1} \otimes \chi_{A_2}^{i_1i_2} \otimes \cdots \otimes  \chi_{A_j}^{i_1i_2 \cdots i_j}$ to the states $\rho^{i_1\cdots i_j}_{A_{j+1}\cdots A_N}$ and $\sigma^{i_1\cdots i_j}_{A_{j+1}\cdots A_N}$. Since both operations are CPTP maps, this keeps the estimates valid due to $\eqref{norm}$. Then we convert all the bounds into the norm $||*||_{A_1:A_2:\cdots :A_3}$ using~\eqref{norm2}. Finally, we combine all the estimates with help of the triangle inequality.
\end{proof}

A first corollary is obtained by combining the theorem with the quantum de Finetti theorem from \cite{CKMR07}. From it, we recover the first part of the statement~\cite[Theorem 1]{DPS04} when all $k_i$ approach infinity~\footnote{The uniqueness of the probability measure (second part of the statement) does not hold in general in this formulation because there exist different decompositions of a separable quantum state. The use of uniqueness of the de Finetti theorem is for the same reason also not applicable in the proof of the first part of the statement.}.

\begin{corollary}
For all $\rho\in  \cE_{A_1:A_2: \cdots :A_N}^{k_1, \ldots, k_{N}}$ there exists $\sigma \in \cS_{A_1:A_2: \cdots :A_N}$ with
$$||\rho-\sigma ||_1\leq 4 \sum_{i=1}^{N-1} \frac{|A_i|}{k_i} .$$
In particular, for $N$ systems of identical dimension $d$, there is an algorithm for deciding separability up to error $\eps$ in the trace norm running in time $\exp(O(d N \log \frac{N}{\eps}))$.
\end{corollary}
\begin{proof}
The claim follows from the bound $\delta(|A|, |B|, k)=4\frac{|A|}{k} $ in trace norm obtained in~\cite[Theorem II.8']{CKMR07}. In order to investigate a bound on the runtime of the algorithm that searches for symmetric extensions, fix an error $\eps>0$, in other words, choose $k=\lceil \eps^{-1} (N-1) d \rceil $. The dimension of the space in which we search for an extension then has dimension 
\begin{align*}
d |\Sym^k(\complex^d)|^{N-1}& \leq d (k+1)^{d-1} \leq \exp( O(d N \log \frac{dN}{\eps})).
\end{align*} 
Since the search for an extension is a semidefinite programme running in time polynomial in the number of variables (which is quadratic in the dimension), we obtain the claimed time complexity (cf.~\cite{BCY10b} for more details in the bipartite case).
\end{proof}

A second corollary of Theorem~\ref{thm:main} concerns the usefulness of combining the search for symmetric extensions with the PPT test. For this let  $\cE_{A_1:A_2: \cdots :A_N}^{k_1, k_2 \ldots, k_{N}, PPT}$ denote the set of states that have a symmetric extension which is furthermore PPT with respect to any partition of the particles into two groups \footnote{In fact, for our purpose it will suffice to demand that the extension be PPT across the cuts $\Sym^{k_i}(A_i):\prod_{j=i+1}^N\Sym^{k_i}(A_i)$ and $\Sym^{k_i/2}(A_i):\Sym^{k_i/2}(A_i)\prod_{j=i+1}^N\Sym^{k_i}(A_i)$.}.
\begin{corollary}
For all $\rho\in  \cE_{A_1:A_2: \cdots :A_N}^{k_1, k_2 \ldots, k_{N}, PPT}$ there exists $\sigma \in \cS_{A_1:A_2: \cdots :A_N}$ with
$$||\rho - \sigma ||_1\leq O\left(\sum_{i=1}^{N-1} \frac{|A_i|^2}{k_i^2} \right). $$
\end{corollary}
\begin{proof}
We wish to use the two-particle bound of Navascues et al.~\cite{NOP09a, NOP09b} in the proof of Theorem~\ref{thm:main}. In order to do so, we need to assert that the states $\rho^{i_1i_2 \cdots i_j}$ that are being constructed with help of the construction of \cite{NOP09a, NOP09b} are again PPT so that we can use their bound again. But this is guaranteed since by inspection of the proof in \cite{NOP09a, NOP09b} the states are obtained by post-selection and hence retain their property of being PPT since the extension is PPT across the cut $\Sym^{k_i}(A_i):\prod_{j=i+1}^N\Sym^{k_i}(A_i)$.
\end{proof}

The following theorem is an alternative to Theorem~\ref{thm:main} tailored to two-particle bounds that only depend on the dimension of $B$. It will be particularly useful when considering the LOCC and Frobenius norms.

\begin{theorem} \label{thm:main2}
Let $||*||$ be a norm that satisfies \eqref{norm} and \eqref{norm2}, assume that $\delta(|A|, |B|, k)$ is a two-particle bound for $||*||$ and set $(k_1, k_2, \cdots k_{N-1},1):=(\ell_1 \ell_2 \cdots \ell_{N-1},  \ell_2  \ell_3 \cdots \ell_{N-1}, \cdots, \ell_{N-1}, 1)$. For all $\rho\in  \cE_{A_1:A_2: \cdots :A_N}^{k_1, \ldots, k_{N}}$ there exists $\sigma \in \cS_{A_1:A_2: \cdots :A_N}$ with
$$||\rho -\sigma ||_{A_1:A_2: \cdots :A_N}\leq \sum_{i=1}^{N-1} \delta \left(\infty, |A_{i+1}| , \ell_i\right).$$
\end{theorem}
\begin{proof}
By assumption there exists an extension
$\rho_{\Sym^{k_1}(A_1)\Sym^{k_1}(A_2) \cdots \Sym^{k_{N-1}}(A_{N-1})A_N}$  of  $\rho_{A_1A_2\cdots A_N}$. Since clearly $\rho_{B_{N-1}A_N} \in \cE^{\ell_{N-1}, 1}_{B_{N-1}A_N}$, where 
\begin{align*}
B_{N-1}:= \Sym^{k_1/\ell_{N-1}} &(A_1)\Sym^{k_2/\ell_{N-1}}(A_2) \\
& \cdots \Sym^{k_{N-2}/\ell_{N-1}}(A_{N-2})A_{N-1} ,
\end{align*}
there exists a state
$$\sigma_{B_{N-1}A_N}=\sum_{i_{N-1}} p(i_{N-1}) \rho^{i_{N-1}}_{B_{N-1}}\otimes \sigma^{i_{N-1}}_{A_N}$$
satisfying
$$||\rho_{B_{N-1}A_N}-\sigma_{B_{N-1}A_N}||\leq \delta(\infty, |A_N|, \ell_{N-1}).$$
We then repeat the same argument for the states $ \rho^{i_{N-1}}_{B_{N-1}}$ thereby decoupling system $A_{N-1}$ from 
\begin{align*}
B_{N-2}:= & \Sym^{k_1/(\ell_{N-2}\ell_{N-1})} (A_1)\Sym^{k_2/(\ell_{N-2}\ell_{N-1})}(A_2) \\
& \qquad \cdots \Sym^{k_{N-3}/(\ell_{N-2}\ell_{N-1})}(A_{N-3})A_{N-2}. 
\end{align*}
We continue this way until we have decoupled $A_2$ from $B_1:=A_1$. We then combine all the estimates with help of the triangle inequality and properties \eqref{norm} and \eqref{norm2} which proves the claim.
\end{proof}

The following corollary of Theorem~\ref{thm:main2} 
shows that detecting multiparticle separability is much more efficient than what was previously anticipated. 
\begin{corollary}
Set $(k_1, k_2, \cdots k_{N-1},1):=(\ell_1 \ell_2 \cdots \ell_{N-1},  \ell_2  \ell_3 \cdots \ell_{N-1}, \cdots, \ell_{N-1}, 1)$; then for all $\rho\in  \cE_{A_1:A_2: \cdots :A_N}^{k_1, \ldots, k_{N}}$ there exists $\sigma \in \cS_{A_1:A_2: \cdots :A_N}$ with
$$||\rho-\sigma ||_{LOCC(A_1: A_2: \cdots :A_N)}\leq \frac{1}{8  \ln 2} \sum_{i=1}^{N-1} \sqrt{\frac{\log |A_i|}{\ell_i}},$$ 
and 
$$||\rho-\sigma ||_{2}\leq \frac{\sqrt{153}}{8 \ln 2} \sum_{i=1}^{N-1} \sqrt{\frac{\log |A_i|}{\ell_i}}.$$ 
In particular, deciding separability up to error $\eps$, in LOCC or Frobenius norm, can be done in time $\exp(O(\eps^{-2(N-1)}N^{2N-1}\prod_{i=1}^N \log |A_i| ))$, i.e.~in quasipolynomial-time for constant error and a constant number of parties. 
\end{corollary}
\begin{proof}
The claim follows from the bound $\delta(|A|, |B|, k)=\frac{1}{8 \ln 2} \sqrt{\frac{\log |B|}{k}} $ in LOCC norm obtained in~\cite[Corollary 2]{BCY10a}. Since the Frobenius norm does not satisfy \eqref{norm}, we carefully go through the proof with the LOCC norm, trace out where needed and then replace all norms by Frobenius norms using the following identity between them, proved in \cite{MWW09}: $\Vert * \Vert_{LOCC} \geq \frac{1}{\sqrt{153}} \Vert * \Vert_2$.

In order to investigate a bound on the runtime of the algorithm that searches for symmetric extensions, fix an error $\eps>0$. Setting $\ell_i:= \frac{1}{(8 \ln 2)^2} (N-1)^2\epsilon^{-2}\log |A_{i+1}|$ we see that the error is at most $\eps$.  The runtime then equals a polynomial in the number of variables, which is smaller than  $|A_N| |A_{N-1}|^{\ell_{N-1}} \cdots |A_{1}|^{\ell_{1}\ell_{2}\cdots \ell_{N-1}}=\exp(O((N-1)^{2N-1} \epsilon^{2(N-1)}\prod_{i=1}^N \log |A_i|))$
since $|\Sym^k(A)|\leq |A|^k$\footnote{The latter is a good bound for large $|A|$, a regime in which we are interested.}. 
\end{proof}

\vspace{0.3 cm}
\noindent \textbf{Recovering a quantum de Finetti theorem:} The final result shows that one can obtain a quantum de Finetti theorem from the bounds derived above. 
\begin{theorem}
\label{thm:defin}
Let $n\leq N \leq k$ and $||* ||$ be a norm satisfying \eqref{norm}, \eqref{norm2} and $||* ||\leq ||* ||_1$. For all permutation-invariant states $\rho_{A^{k}}$ \footnote{A state is permutation-invariant if for all permutations $\pi$ of $k$ elements $U_{\pi} \rho_{A^{k}} U_{\pi}^\dagger =\rho_{A^{k}}$.} there exists a state $\sigma_{A^k}=\sum_i p_i \sigma_i^{\otimes k} $ with 
$$||\rho_{A^n} -\sigma_{A^k}||_{A:A:\cdots :A}\leq (N-1) \delta (\infty, |A|, k^{\frac{1}{N}})+2\frac{n^2}{N}$$
\end{theorem}
\begin{proof} 
We first apply Theorem~\ref{thm:main2} and subsequently \cite[Theorem 6]{TC09}.
\end{proof}

A direct consequence of the previous theorem is a de Finetti theorem in LOCC norm which only depends \textit{logarithmically} on the local dimension. 
\begin{corollary}
Let $n\leq N \leq k$.  For all permutation-invariant states $\rho_{A^{k}}$ there exists a state $\sigma_{A^k}=\sum_i p_i \sigma_{A, i}^{\otimes k} $ with 
$$||\rho_{A^k}-\sigma_{A^k}||_{\LOCC(A: A: \cdots :A)}\leq (N-1) \frac{\sqrt{\log |A|} }{k^{\frac{1}{2N}}}+2\frac{n^2}{N}$$
which goes to zero for $n << N << k$ (e.g. for constant $n$, $N=\sqrt{\log k}$ and $k\rightarrow \infty$).
\end{corollary}
\begin{proof} The claim follows directly from Theorem~\ref{thm:defin} and \cite[Corollary 2]{BCY10a}.
\end{proof}

\vspace{0.3 cm}
\noindent \textbf{Discussion:}  Fast algorithms for deciding separability of quantum states are important both from a theoretical perspective in quantum information theory and from the point of view of experimental work in the generation of multiparticle quantum states. There one is typically interested in showing that the state created is entangled, as a proof of a successful implementation of the experiment. It is therefore of great interest to devise better algorithms for separability and to understand the strentghs and limitations of current techniques. In this work we have shown that the detection of multiparticle entanglement can be done much faster than what was previously thought. Our result gives conclusive evidence for the use of the Search for Symmetric Extensions as a practical test of separability and we hope it fosters further theoretical and experimental investigation on the detection of multiparticle entanglement.

\acknowledgements
FB is supported by a ``Conhecimento Novo" fellowship from FAPEMIG. MC is supported by the Swiss National Science Foundation (grant PP00P2-128455) and the German Science Foundation (grants~CH 843/1-1 and CH 843/2-1). 

\end{document}